\documentclass[11pt,letterpaper,preprintnumbers,superscriptaddress]{revtex4}

\usepackage{fullpage, dsfont, amsthm, amssymb, amsmath}
\usepackage{graphicx,color,tikz}
\usepackage{hyperref}

\begin{document}

\title{Adiabatic Optimization Versus Diffusion Monte Carlo}

\newcommand{\quics}{Joint Center for Quantum Information and Computer Science, University of Maryland}

\author{Michael Jarret}
\affiliation{\quics}
\affiliation{Department of Physics, University of Maryland, College Park}
\author{Stephen P. Jordan}
\affiliation{\quics}
\affiliation{National Institute of Standards and Technology, Gaithersburg, MD}
\author{Brad Lackey}
\affiliation{\quics}
\affiliation{Department of Mathematics, University of Maryland, College Park}
\affiliation{National Security Agency, Ft. G. G. Meade, MD}

\date{\today}

\newcommand{\id}{\mathds{1}}
\newcommand{\eq}[1]{(\ref{#1})}
\newcommand{\bra}[1]{\langle #1|}
\newcommand{\ket}[1]{|#1\rangle}
\newcommand{\braket}[2]{\langle #1|#2\rangle}
\renewcommand{\th}{^\mathrm{th}}

\newtheorem{proposition}{Proposition}
\newtheorem{hypothesis}{Hypothesis}
\newtheorem{lemma}{Lemma}

\begin{abstract}
Most experimental and theoretical studies of adiabatic optimization
use \emph{stoquastic} Hamiltonians, whose ground states are
expressible using only real nonnegative amplitudes. This raises a
question as to whether classical Monte Carlo methods can simulate
stoquastic adiabatic algorithms with polynomial overhead. Here, we
analyze diffusion Monte Carlo algorithms. We argue that, based on
differences between $L_1$ and $L_2$ normalized states, these algorithms
suffer from certain obstructions preventing them from efficiently
simulating stoquastic adiabatic evolution in generality. In practice
however, we obtain good performance by introducing a method that we
call Substochastic Monte Carlo. In fact, our simulations are good
classical optimization algorithms in their own right, competitive with
the best previously known heuristic solvers for MAX-$k$-SAT at
$k=2,3,4$.
\end{abstract}

\maketitle

\section{Introduction}

While adiabatic quantum computation using general Hamiltonians has
been proven to be universal for quantum computation \cite{ADKLLR07},
the vast majority of research so far, both experimental and
theoretical, focuses on Hamiltonians in which all off-diagonal matrix
elements are nonpositive. Such Hamiltonians were named
\emph{stoquastic} in \cite{BDOT08}. By the Perron-Frobenius theorem,
the ground state of a stoquastic Hamiltonian can always be expressed
using only real nonnegative amplitudes. Consequently, in adiabatic
computations, which stay in the ground state, interference
effects are not manifest if the Hamiltonian is stoquastic. This raises
some question as to whether adiabatic computation in the ground state
of stoquastic Hamiltonians is capable of exponential speedup over
classical randomized algorithms. Complexity-theoretic evidence
obtained so far suggests that adiabatic quantum computation with
stoquastic Hamiltonians is less powerful than universal quantum
computers \cite{BDOT08, BT09} but does not resolve this
question. Conventional wisdom among Monte Carlo practitioners states
that Monte Carlo simulations of stoquastic adiabatic computation will
not suffer from the sign problem and will therefore converge
efficiently. If this could be turned into a theorem it would prove
that stoquastic adiabatic computers are incapable of exponential
speedup over classical computation.

Two of the major classes of Monte Carlo simulation algorithms are path
integral Monte Carlo and diffusion Monte Carlo. In 2013, Hastings
constructed a class of examples in which path integral Monte
Carlo fails to efficiently simulate stoquastic adiabatic dynamics
due to topological obstructions \cite{Hastings}. Diffusion Monte Carlo
algorithms should not be affected by topological obstructions. We
nevertheless find examples for which a nontopological
obstruction prevents a diffusion Monte Carlo algorithm from efficiently simulating
a stoquastic adiabatic process. On typical instances such simulations
may nevertheless work well. In fact, we introduce a variant of
diffusion Monte Carlo, which we call Substochastic Monte Carlo (SSMC),
tailored to simulating stoquastic adiabatic processes. In practice, we
find that this performs sufficiently well that SSMC
simulations of adiabatic optimization are good classical optimization
algorithms in their own right, competitive with the best previously
known heuristic solvers for MAX-$k$-SAT at $k=2,3,4$. Source code for
our implementation is available at \cite{github}.

The relationship of quantum adiabatic optimization
\cite{Farhi_science} and quantum annealing \cite{FGSSD94} to classical
optimization heuristics and simulation methods has garnered a lot of
attention. In particular, the most direct classical competitors to
quantum adiabatic optimization appear to be gradient descent and
simulated annealing, path integral Monte Carlo, and diffusion Monte
Carlo. One can analytically address the performance of adiabatic optimization
algorithms through adiabatic theorems \cite{JRS07, Elgart_Hagedorn},
which show that the runtime of adiabatic algorithms corresponding to
Hamiltonians with eigenvalue gap $\gamma$ is upper bounded by
$O(1/\gamma^2)$. Examples have been constructed in which adiabatic
optimization exponentially outperforms simulated annealing and
gradient descent \cite{FGG02}. Conversely, problems exist that can be
solved in polynomial time by gradient descent but which have a
corresponding exponentially small eigenvalue gap \cite{JJ15}. In the
examples of \cite{Boixo2, Boixo3, Boixo4}, runtimes for path integral
Monte Carlo and quantum adiabatic optimization have the same
asymptotic scaling. However, the obstructions of \cite{Hastings} show
that there exist cases where adiabatic optimization exponentially
outperforms path integral Monte Carlo. Other informative analytical
results on the effect of local minima on the performance of adiabatic
optimization are given in \cite{R04, DMV01, VDV03, Amin, aminchoi,
  Boixo1, BvD16, CH16}. Experimental, numerical, and analytical
evidence regarding the performance of quantum adiabatic optimization
on combinatorial optimization problems such as MAX-SAT can be found in
\cite{Farhi_science, AKR10, diffinit, McGeoch, Dwave}. Some
variants of the standard adiabatic optimization algorithm involving
non-linear interpolation or non-stoquastic Hamiltonians are analyzed
in \cite{diffmid2, diffmid1}.

\section{Terminology and Notation}

Let $G = (V,E)$ be a graph with vertex set $V$ and edge set $E
\subseteq V \times V$. We presently restrict our attention to
unweighted undirected graphs without self-loops. (Generalizing to
weighted graphs is easy, however uninstructive for our current
purposes.) By $d_j$ we denote the degree of vertex $j \in V$, that is,
the number of edges with an endpoint $j$. We use $L^{(G)}$ to denote
the combinatorial Laplacian of $G$, which is a $|V| \times |V|$ matrix
given by
\begin{equation}
L^{(G)}_{ij} = \left\{ \begin{array}{rl}
-1 & \textrm{if $(i,j) \in E$} \\
d_i & \textrm{if $i = j$} \\
0 & \textrm{otherwise.}
\end{array} \right.
\end{equation}
As an example, consider the $n$-dimensional hypercube graph $H_n$. The
$2^n$ vertices of $H_n$ may be labeled by bitstrings, and an edge
connects a pair of bitstrings if and only if they differ by one
bit. We write the corresponding combinatorial Laplacian in terms of the
Pauli operators as
\begin{equation}
\label{HyperPauli}
L^{(H_n)} = n \id - \sum_{k=1}^n X_k.
\end{equation}
Here and throughout we use $X_k,Y_k,Z_k$ to denote the Pauli $x,y,z$
operators acting on qubit $k$ (tensored with identity on the remaining
qubits). Aside from the (inconsequential) energy offset $n\id$, the
Laplacian $L^{(H_n)}$ is the most commonly used driving term for
adiabatic optimization algorithms.

\section{Substochastic Monte Carlo}

By Substochastic Monte Carlo (SSMC) we denote the class of classical
population algorithms that simulate a time-dependent heat diffusion
process driven by the same operator as the stoquastic adiabatic
process. The particular variant of SSMC studied below can be viewed as
either a form of diffusion Monte Carlo \cite{Grimm_Storer} or as a
particular generalization of the ``Go-With-The-Winners'' algorithm of 
\cite{AV94}. Consider a parametrized family $H(s)$ of stoquastic Hamiltonians,
which define an adiabatic algorithm by slowly varying $s$ from 0
to 1 according to some schedule $s(t)$. The corresponding
imaginary-time dynamics
\begin{equation}
\label{diffueq}
\frac{d}{dt} \psi = - H(s(t)) \psi
\end{equation}
is a continuous time diffusion process. Correspondingly, at any given
$s$, the time-evolution operator $e^{-H(s)\Delta t}$ is a
substochastic matrix for any sufficiently small positive time step
$\Delta t$. Such substochastic matrices (\emph{i.e.} square matrices with
  nonnegative entries such that each column sum is at most 1) generate
substochastic random processes, in which total probability
decreases. The lost probability is taken to represent the chance
that the process stops or ``dies''. SSMC is a Markov-Chain Monte Carlo
method in which a population of random-walkers approximates the
diffusion process on the graph given by this substochastic process
conditioned on not having died.

While SSMC works with any family of stoquastic Hamiltonians, it is
easiest to describe for 
$$H(s) = a(s)L + b(s)W,$$
where $L$ is a combinatorial graph Laplacian, $W =
\mathrm{diag}\{w_1,w_2,w_3,\ldots\}$ is a nonnegative diagonal
operator, and $a(s)$ and $b(s)$ are suitable nonnegative 
scalar functions. When this is the case we may write 
$$H(s) = a(s)(D - A) + b(s)W,$$
where $D$ is the diagonal operator of vertex degrees, and $A$ is the
adjacency matrix of the graph. For a sufficiently small time step we
approximate
$$e^{-H(s)\Delta t} \approx \id - H(s)\Delta t = (\id - a(s)\Delta t D - b(s)\Delta t W)
+ a(s)\Delta t A$$
which prescribes our transition probabilities. At time $t$,
the value $s = s(t)$ is computed according to the schedule, and a walker
on vertex $j$ would do precisely one of the following:
\begin{enumerate}
\item step to vertex $i\in V$ (where $(i,j)\in E$), each with probability $a(s)\Delta t$,
\item stay at vertex $j$ with probability $1 - a(s)\Delta t d_j - b(s)\Delta t w_j$,
\item or die with probability $b(s)\Delta t w_j$.
\end{enumerate}

The expected proportion of walkers that die in a given timestep is 
$b(s)\Delta t \langle W \rangle_t$ where the expectation is computed with 
respect to the current ensemble population distribution at time $t$.
The proportion of the walkers that survive is $1 - b(s)\Delta t
\langle W \rangle_t \approx e^{-b(s)\langle W\rangle_t\Delta
  t}$. Conditioning on survival renormalizes our ensemble
distribution, multiplying by $e^{b(s)\langle W\rangle_t\Delta
  t}$. Combining this with the above produces our combined
transition/renormalization matrix 
\begin{equation}\label{eqn:transition}
e^{b(s)\langle W\rangle_t\Delta t}\cdot e^{-H(s)\Delta t} 
\approx \left[\id - a(s)\Delta t D - b(s)\Delta t (W - \langle W \rangle_t)\right] + a(s)\Delta t A.
\end{equation}

If $s$ were constant in $t$, this substochastic process converges
to a quasistationary distribution (\emph{i.e.} a distribution that is
stationary except for exponentially decaying norm) that is
proportional to the ground state of the original stoquastic
Hamiltonian. Thus, one can attempt to simulate stoquastic adiabatic
evolution using a substochastic classical random walk. The simplest
idea would be to initialize the walkers into the ground state
distribution of $H(0)$, which is typically the uniform distribution,
and then track the instantaneous quasistationary distribution, as $s$
slowly increases from zero to one, by executing the Markov chain 
$\prod_j \left( \id - H(s(t_j)) \Delta t_j \right)$. (Here $\Delta t_j$ is the size of
the $j^\mathrm{th}$ timestep and $t_j = \sum_{k=1}^j \Delta t_j$.) However, one
needs to introduce some process for replenishing the population of
walkers. Otherwise, after a short time there are no walkers left and
the simulation terminates. 

There are a number of potential ways to replenish the
walkers. We have found it effective to adaptively set an
energy threshold throughout the time evolution such that walkers on
sites with energy above the threshold are likely to die, whereas
walkers on sites with energy below the threshold are likely to spawn
offspring. According to \eq{eqn:transition} that threshold should be 
the mean energy of the population, $\langle W\rangle_t$. Specifically, 
in our scheme, at each timestep, a walker on vertex $j$:
\begin{enumerate}
\item steps to vertex $i\in V$ (where $(i,j)\in E$), each with probability $a(s)\Delta t$,
\item stays at vertex $j$ with probability 
$1 - a(s)\Delta t d_j - |b(s)\Delta t (w_j - \langle W \rangle_t)|$, or
\item dies or spawns a new walker based on remaining probability 
$|b(s)\Delta t (w_j - \langle W\rangle_t)|$.
\end{enumerate}
We must have $0 \leq a(s) \Delta t d_j \leq 1$ for the probabilities
in Case 1 to make sense. A similar statement holds for Case 2, from
which we derive 
$$-a(s)\Delta t d_j \leq b(s)\Delta t (w_j - \langle W\rangle_s) \leq 1 - a(s)\Delta t d_j.$$
In particular, $b(s)\Delta t (w_j - \langle W\rangle_s) \in [-1,1]$ and we interpret Case 3 to be
\begin{enumerate}
\item[3a.] if $b(s)\Delta t (w_j - \langle W\rangle_s) > 0$ the walker dies with this probability, or
\item[3b.] if $b(s)\Delta t (w_j - \langle W\rangle_s) < 0$ then with probability $b(s)\Delta t (\langle W\rangle_s - w_j)$ the walker spawns an additional walker at vertex $j$.
\end{enumerate}
This choice of probabilities for spawning or dying ensures that the
quasistationary distribution is proportional to the ground state of
$H(s)$.

Note that the population size is itself a random variable. In theory,
the threshold between dying and spawning is $\langle W\rangle_t$, however
in practice one must adjust this to ensure the population
size stays sufficiently close to a nominal value. This
can be accomplished by introducing a feedback loop, which replaces
$\langle W\rangle_s$ with $\langle W\rangle_s - E$ for some energy
offset adaptively chosen based on the number of walkers. When the
population size dwindles below the target value, $E$ is
decreased. As one can see by examining the formulas defining 1, 2, 3a, 
and 3b, this replacement increases the likelihood for walkers to
spawn, thereby replenishing the population. Conversely, when the
number of walkers increases beyond the target population size, $E$ is
increased, thereby increasing the likelihood for walkers to die. 

Substochastic Monte Carlo can be viewed either as a method for
simulating stoquastic adiabatic computation or as a method for solving
discrete optimization problems. In the latter case, it is natural to
ask why $s$ must be varied at all. In the typical case, $a(1) = 0$ and
$b(1) = 1$. Thus, $H(1) = W = \mathrm{diag}\{w_1,w_2,w_3,\ldots\}$
where $w_1,w_2,\ldots$ is the objective function that we seek to
minimize. In this case, the pure diffusion process $\frac{d}{dt} \psi = -
H(1) \psi$ converges rapidly to the minimum energy state. However, in
typical problems a 
good approximation to this diffusion process can generally only be
achieved using an exponentially large population of walkers. A diagonal 
Hamiltonian $H(1) = W$ implies that the probability for a walker to hop between vertices in
SSMC is zero. The only remaining processes are death and spawning. If
at least one walker is sitting at a minimum energy site, then
death and spawning guarantee that the entire population converges to these sites, in agreement with the diffusion equation. 
However, in an optimization problem one does not initially know the optimum and therefore the initial distribution of walkers cannot depend upon knowledge of the minimum energy. For example, if
the problem has a unique minimum energy vertex, the uniform
distribution over all $2^n$ vertices has exponentially small overlap
with the quasistationary distribution of $H(1)$, which is supported
entirely by a single vertex. In this case, with only polynomially many
walkers,  it is exponentially unlikely that any walker is initially
placed at the solution, and, due to lack of hopping, no walker
arrives at the solution.

Intuitively, when SSMC is applied to optimization problems, the
sweeping of $H(s)$ from the graph Laplacian $L^{(G)}$
at $s=0$ to the diagonal matrix $W$ at $s=1$ serves a role loosely
analogous to decreasing temperature in simulated annealing. Initially, when $s$ is small, the population of walkers
explores widely. As $s$ is increased, the walkers become gradually
more focused around the regions of the search space where the
objective function has been found to take small values. If SSMC
successfully tracks the quasi-stationary distribution then, after a
given timestep, the walkers are distributed close to the
quasi-stationary distribution of the current Hamiltonian $H(s)$. This
then serves as the initial distribution for the next timestep with
Hamiltonian $H(s+\Delta s)$. If $s$ is varied sufficiently slowly,
then this initial distribution is close to the quasi-stationary
distribution of $H(s+ \Delta s)$, which facilitates convergence to the
new quasistationary distribution.

If the substochastic Monte Carlo simulation successfully simulates the
adiabatic process, then the final distribution of walkers will be
proportional to the final ground state, which in the case of $H(1) =
W$ has support only on the minimum energy vertex. For solving an
optimization problem defined by $W$, this is sufficient though
overkill; the optimum is found if at least one walker lands on the
minimum energy vertex.

\section{Non-Topological Obstructions}

In this section we present a pair of stoquastic adiabatic
processes which diffusion Monte Carlo algorithms such as SSMC will
fail to efficiently simulate (with a stronger notion of failure
in the second, more elaborate, example). Previously, \cite{Hastings}
gave examples of stoquastic adiabatic processes that have polynomial
eigenvalue gap but path integral Monte Carlo simulations of these
processes take exponential time to converge. Loosely speaking, the
failure of convergence was due to topological obstructions around
which the worldlines can get tangled. In diffusion Monte Carlo
algorithms, such as SSMC, there are no world lines, and
correspondingly no susceptibility to these topological
obstructions. Instead, our examples exhibit a different kind of
obstruction, exploiting the fact that diffusion Monte Carlo
simulations track the probability distribution proportional to the
ground state amplitudes rather than the squared amplitudes. Our
examples are inspired by the fourth counterexample given in \S 3.4 of
\cite{Hastings}, in which a discrepancy between the $L_1$ and
$L_2$-normalized wavefunctions is exploited to demonstrate exponential
convergence time for a path integral Monte Carlo simulation with open
boundary conditions.

For $s \in [0,1]$, let $H(s)$ be some stoquastic Hamiltonian acting on
a Hilbert space whose basis states can be equated with the vertices
$V$ of some graph. Let $\psi_s(x):V \to \mathbb{C}$ denote the ground state
of $H(s)$. Diffusion Monte Carlo algorithms (including SSMC) perform random walks designed to ensure that a population
of random walkers converges to the probability distribution
$p^{(1)}_s$ on $V$ directly proportional to the ground state
$\psi_s(x)$. That is, 
\begin{equation}
p^{(1)}_s(x) = \frac{\psi_s(x)}{\sum_{y \in V} \psi_s(y)}.
\end{equation}
The stoquasticity of $H(s)$ ensures that $\psi_s(x)$ is always real
and nonnegative, and consequently that $p^{(1)}_s$ is a valid
probability distribution. In contrast, the probability distribution
sampled from by performing a measurement on the quantum ground state
of the adiabatic process is
\begin{equation}
p^{(2)}_s(x) = \psi_x(x)^2.
\end{equation}

In exponentially large Hilbert spaces there can be vertices such that
$p^{(2)}_s(x)$ is polynomial but $p^{(1)}_s(x)$ is exponentially
small. The idea behind our examples is to exploit this discrepancy to
design polynomial-time stoquastic adiabatic processes that the
corresponding diffusion Monte Carlo simulations will fail to
efficiently simulate.\\
\\
\textbf{Example 0:} Consider the hypercube on $n$ qubits, and let $L$
be the hypercube graph Laplacian, as described in
\eq{HyperPauli}. Consider the stoquastic adiabatic Hamiltonian
\begin{equation}
\label{H0}
H_0(s) = \frac{1}{n} \left[ L + sbW \right]
\end{equation}
where $W$ is the Hamming weight potential. That is,
\begin{equation}
W \ket{x} = |x| \ \ket{x} \quad \textrm{for $x \in \{0,1\}^n$},
\end{equation}
where $|x|$ denotes the Hamming weight of bit string $x$, \emph{i.e.}
the number of ones. (In terms of Pauli
operators $W = \sum_{j=1}^n (\id - Z_j)/2$.) By the straightforward
calculation given in Appendix \ref{calc}, one finds that by choosing
\begin{equation}
\label{bval}
b = \frac{2}{\tan \left[ 2 \cos^{-1} \left( 1 - \frac{1}{4n} \right) \right]}
\end{equation}
one obtains a ground state probability distribution $p^{(2)}_{s=1}$ that has
\begin{equation}
\label{pzero}
p^{(2)}_{s=1}(0\ldots0) = \left( 1 - \frac{1}{4n} \right)^{2n}
\end{equation}
and therefore
\begin{equation}
\lim_{n \to \infty} p^{(2)}_{s=1}(0\ldots0)  = \frac{1}{\sqrt{e}}
\end{equation}
whereas the corresponding distribution of random walkers behaves as
\begin{equation}
p^{(1)}_{s=1}(0\ldots0) = O\left( e^{-\sqrt{n/2}} \right).
\end{equation}

The minimum eigenvalue gap of $H_0(s)$ occurs at $s=0$ and is equal to
$\frac{2}{n}$. Thus, this already constitutes an example where adiabatically
evolving according to $H_0(s)$ with $s$ varying from $0$ to $1$ over a
polynomial duration and then measuring in the computational basis yields
the minimum potential with constant probability, whereas diffusion
Monte Carlo algorithms have subexponentially small probability
of querying the minimum. This is not an especially compelling example,
because SSMC may nevertheless efficiently converge to the probability
distribution $p^{(1)}_s$. That is, Example 0 disproves the naive
hypothesis that if a measurement of an adiabatic process consistently
yields the minimum of a potential in polynomial time, then so does
SSMC. This, however, reflects only the exponential divergence in the
$L_1$ and $L_2$ norms and does not disprove the following more nuanced
hypothesis.

\begin{figure}
\begin{center}
\includegraphics[width=0.25\textwidth]{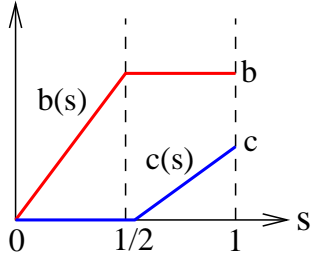}
\caption{\label{sched} The annealing schedule for Example 1
  recapitulates Example 0 from $s=0$ to $s=1/2$. Afterwards, from
  $s=1/2$ to $s=1$, the potential on the all-zeros bitstring is
  lowered by some amount $c$.}
\end{center}
\end{figure}

\begin{figure}
\begin{center}
\includegraphics[width=0.6\textwidth]{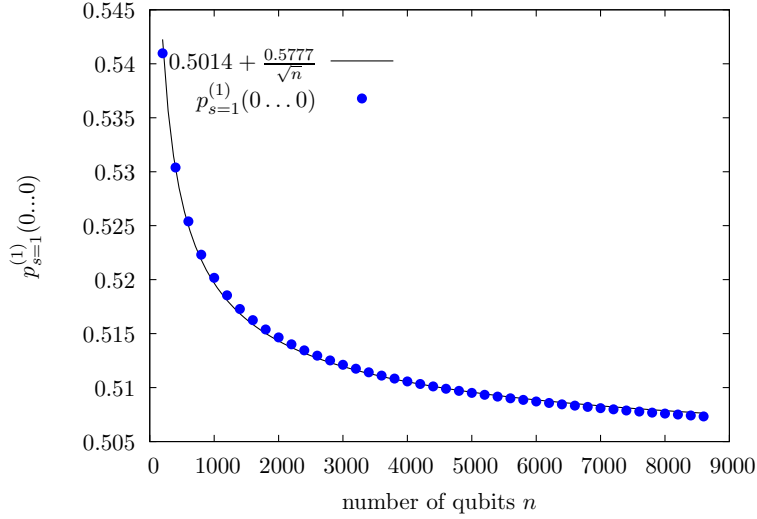}
\caption{\label{numerics} The ground state of $H_1(s)$ lies within the
  subspace of states that are invariant under all permutations of
  qubits. Consequently, the ground state can be obtained numerically
  up to large numbers of qubits, as described in appendix
  \ref{symmetry}. With our choices of $b$
  and $c$, we numerically find that $p_{s=1}^{(1)}(0\ldots0)$ is convergent to a
  constant (approximately equal to 1/2) in the limit $n \to \infty$.}
\end{center}
\end{figure}


\begin{hypothesis}
\label{hyp}
For all $s \in [0,1]$ let $H(s)$ be a stoquastic Hamiltonian with
ground state $\psi_s$ and eigenvalue gap $\gamma(s)$. Let $\gamma =
\min_{0 \leq s \leq 1} \gamma(s)$. There exist polynomials $p,q$ such
that with $p(\gamma, 1/\epsilon)$ timesteps and 
$q(\gamma, 1/\epsilon)$ walkers, SSMC tracks a
probability distribution $\epsilon$-close to $p^{(1)}_s$.
\end{hypothesis}

We can disprove this hypothesis with the following, slightly more
elaborate example.\\
\\
\textbf{Example 1:} Consider the Hamiltonian
\begin{equation}
\label{h1}
H_1(s) = \frac{1}{n} \left[ L + b(s) W \right] - c(s) P
\end{equation}
where $L$ and $W$ are as in Example 0 and $P =
\ket{0\ldots0}\bra{0\ldots0}$ is the projector onto the all zeros
bitstring. The ``annealing schedule'' for $s \in [0,1]$ is given by
\begin{eqnarray}
b(s) & = &\left\{ \begin{array}{ll}
2sb & s \leq 1/2 \label{schedule1}\\
b\phantom{(2s-1)c} & s > 1/2 \label{schedule2}
\end{array} \right.
\\
c(s) & = & \left\{ \begin{array}{ll}
0 & s \leq 1/2 \\
(2s-1)c\phantom{b} & s > 1/2
\end{array} \right.
\end{eqnarray}
as illustrated in Figure \ref{sched}. The constant $b$ is chosen as in
\eq{bval}.

As proven in Appendix \ref{gaproof}
\begin{equation}
\min_{1/2 \leq s \leq 1} \gamma(s) \simeq \frac{1}{\sqrt{2n}},
\end{equation}
for any $c \geq 0$. The spectrum for $s < 1/2$ recapitulates the
spectrum of example 0. Thus, the minimum eigenvalue gap over the full adiabatic
process occurs at $s=0$ and is given by $\gamma =
\frac{2}{n}$. Consequently, the quantum adiabatic implementation of this
process runs in polynomial time. By choosing $c$ sufficiently large we
can ensure that $p^{(1)}_{s=1}(0\ldots0)$ is
$\Omega(1)$. Specifically, one finds numerically that by choosing $c =
2$ one obtains
\begin{equation}
p^{(1)}_{s=1}(0\ldots0) \simeq 0.50 + \frac{0.58}{\sqrt{n}},
\end{equation}
as shown in Figure \ref{numerics}. Thus, to satisfy Hypothesis
\ref{hyp}, the walkers would have to end up at $s=1$ in a probability
distribution with probability approximately $1/2$ 
at the all zeros string. However, from the analysis of Example 0, we
know that at $s=1/2$ the walkers have a distribution in which the
probability to be at the all zeros string is of order
$e^{-\sqrt{n/2}}$. Thus, with high likelihood, no walkers will land on
the all zeros string until the number of timesteps $T$ times the
number of walkers $W$ approaches $TW \sim e^{\sqrt{n/2}}$. Until this
happens it is impossible for the distribution of walkers to be
affected by the change in the potential at the all zeros string that
is occurring from $s=1/2$ to $s=1$; no walkers have landed there, and
the diffusion Monte Carlo algorithm has therefore never queried the
value of the potential at that site.

\section{Empirical Performance of Substochastic Monte Carlo}

\begin{figure}[ht]
\includegraphics[width=0.6\textwidth]{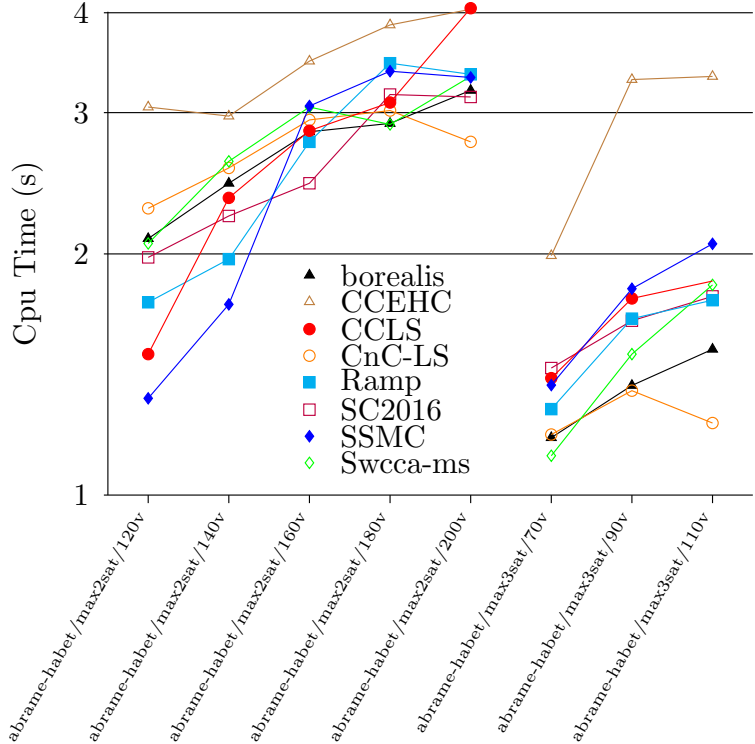}
\caption{Runtime comparison of several solvers from MAX-SAT 2016 contest.}
\label{fig:runtimes}
\end{figure}

\begin{figure}[ht]
\includegraphics[width=0.6\textwidth]{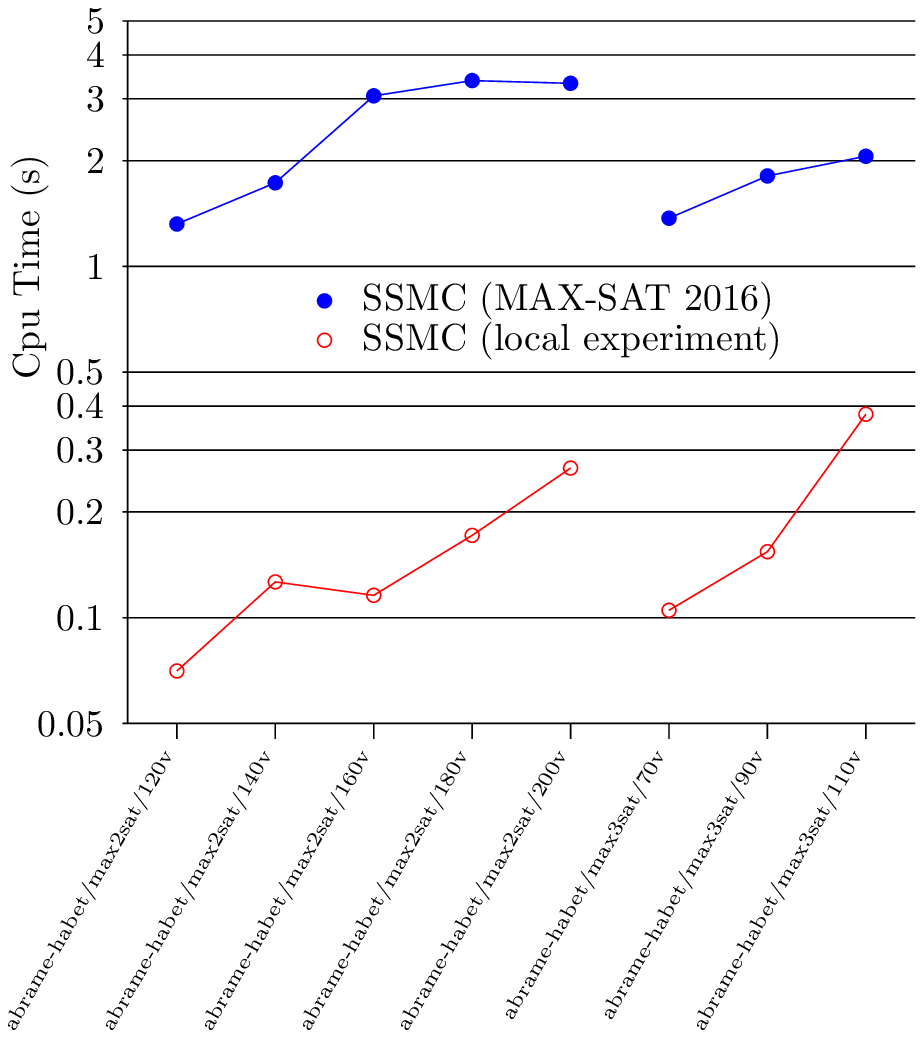}
\caption{Runtime of SSMC in MAX-SAT 2016 contest versus local experiments.}
\label{fig:runtimes2}
\end{figure}

While highly structured problems may lead to obstructions, as in the
previous section, unstructured and random problems are not likely to
see these. There are numerous benchmarks of random problems available;
here we provide results of SSMC and seven other solvers on a selection
of unweighted random MAX-SAT problems from the 2016 MAX-SAT evaluation \cite{MaxSat2016}. We omitted
solvers that did not solve every problem instance in these categories. Also
we do not show results on the high-girth examples of this benchmark as several of the algorithms (including SSMC) did not succeed at finding an
optimal solution for every instance. In Figure \ref{fig:runtimes}, we see comparable timings for all the solvers. There is a general upward trend in the MAX-2-SAT and MAX-3-SAT timings versus number of variables, but it is hard to discern the exponential behavior one would expect for solving MAX-SAT problems.

For SSMC, an exponential runtime was programmed explicitly. We selected $T = e^{0.022n + 5.9}$ for MAX-2-SAT and $T = e^{0.035n + 6.1}$ for MAX-3-SAT. A linear schedule was used, $a(s) = 1-s$ and $b(s) = s$ where $s = t/T$. This selection was based on tuning the parameters so as to maximize the
success rate using a constant number of walkers across all problem instances (namely, sixteen). The bulk of the work is in computing the potential of a walker (\emph{i.e.} the number of failed SAT clauses). An improved implementation of this computation, as well as optimization over the number of walkers and schedule, is expected to yield a better scaling. 

The SSMC contest timings in Figure \ref{fig:runtimes} do not appear consistent with exponential scaling, but a local experiment displays a clear exponential trend of SSMC consistent with our programmed runtimes, Figure \ref{fig:runtimes2}. Precisely the same codebase and benchmark problems were used. There were slight differences in hardware and compiler, locally LLVM 7.3.0 on an Intel Xeon 2.6GHz Mac Pro, while the contest utilized GCC 4.4.7 on an Intel Xeon 2.0GHz CentosOS Linux server. This seems unlikely to be the cause for this discrepancy in behavior, as \emph{borialis} also exhibited contest timings very different than those reported in \cite{zhu2016borealis}. 

We believe that whatever factor of the contest environment distorted
the runtimes of SSMC consistently affected the other
algorithms. Consequently, although we can conclude that SSMC is a
competitive solver for the MAX-SAT evaluation, we cannot confidently
extrapolate scaling from the contest results from which to compare
SSMC with other solvers. In particular, Figure \ref{fig:runtimes}
shows a negative slope between the $180$ and $200$ variable instances
(respectively abrame-habet/max2sat/180v and
abrame-habet/max2sat/200v). With the programmed runtimes, this is
highly unexpected and could not be locally reproduced, but was
consistent with the competition results for other solvers. In future
work, in order to clarify relative scaling behavior and improve the
optimization of SSMC, we plan run several of the solvers submitted to
the MAX-SAT competition against SSMC in a local environment that
captures timings whose interpretation is more clear.

\textbf{Acknowledgments:} We thank Aaron Ostrander for useful
discussions and Yi-Kai Liu for making us aware of
\cite{AV94}. MJ thanks Booz Allen Hamilton for support. Portions of
this paper are a contribution of NIST, an agency of the US government,
and are not subject to US copyright. 

\clearpage

\appendix

\section{Calculations for Example 0}
\label{calc}

We can re-express \eq{H0} as
\begin{eqnarray}
H_0(s) & = & \left( 1 + \frac{sb}{2} \right) \id - \frac{1}{n}
\sum_{j=1}^n \left( X_j + \frac{sb}{2} Z_j \right) \\
& = & \left(1 + \frac{sb}{2} \right) \id - \frac{1}{n} 
\sqrt{1+ \left( \frac{sb}{2} \right)^2} \sum_{j=1}^n \left( 
\sin(\theta) X_j + \cos(\theta) Z_j \right) \label{hthetan}
\end{eqnarray}
where
\begin{equation}
\label{theta}
\theta = \tan^{-1} \left( \frac{2}{sb} \right).
\end{equation}
From \eq{hthetan} one can see that the eigenvalue gap of $H_0(s)$ is
\begin{equation}
\label{rampgap}
\gamma(s) = \frac{2}{n} \sqrt{1 + \left( \frac{sb}{2} \right)^2}.
\end{equation}
Furthermore, the ground state of
\begin{equation}
H(\theta) = - \sin(\theta) X - \cos(\theta) Z
\end{equation}
is
\begin{equation}
\ket{\psi(\theta)} = \cos \left( \frac{\theta}{2} \right) \ket{0} +
\sin \left( \frac{\theta}{2} \right) \ket{1}.
\end{equation}
So, the ground state of $H_0(s)$ is $\ket{\psi(\theta)}^{\otimes n}$
with $\theta$ given by \eq{theta}. We choose $b$ so that at $s=1$ we
have $\cos(\theta/2) = 1-\frac{1}{4n}$. By \eq{theta} this entails
\begin{equation}
b = \frac{2}{\tan \left[ 2 \cos^{-1} \left( 1 - \frac{1}{4n} \right) \right]}.
\end{equation}
With this choice of $b$ we have
\begin{equation}
p^{(2)}_{s=1}(0\ldots0) = \psi_{s=1}^2(0\ldots0) = \cos \left(
  \frac{\theta}{2} \right)^{2n} = \left( 1 - \frac{1}{4n} \right)^{2n}
\end{equation}
which is asymptotically a constant, specifically converging to
$1/\sqrt{e}$ as $n \to \infty$. 

Now, consider the probability distribution sampled by the diffusion
Monte Carlo algorithm.
\begin{equation}
p^{(1)}_s(x) = \frac{1}{\mathcal{Z}_s} \sin(\theta/2)^{|x|} \cos(\theta/2)^{n-|x|}
\end{equation}
where $|x|$ is the Hamming weight if $x \in \{0,1\}^n$ and
\begin{eqnarray}
\mathcal{Z}_s & = & \sum_{x \in \{0,1\}^n} \sin(\theta/2)^{|x|} \cos(\theta/2)^{n-|x|} \\\
& = & \sum_{w = 0}^n \binom{n}{w} \sin(\theta/2)^w \cos(\theta/2)^{n-w} \\
& = & \left[ \sin(\theta/2) + \cos(\theta/2) \right]^n.
\end{eqnarray}
At $s=1$, $\cos(\theta/2) = 1-\frac{1}{4n}$ and $\sin(\theta/2) =
\sqrt{1-\left(1-\frac{1}{4n} \right)^2}$. So, for large $n$
\begin{eqnarray}
\mathcal{Z}_{s=1} & \simeq & \left[ 1 + \sqrt{\frac{1}{2n}} \right]^n \\
& = & e^{n \log ( 1 + 1/\sqrt{2n})} \\
& \simeq & e^{\sqrt{n/2}}.
\end{eqnarray}
Thus,
\begin{equation}
p^{(1)}_{s=1}(0\ldots0) = 
\frac{\cos(\theta_{s=1}/2)^n}{\mathcal{Z}_{s=1}} = \frac{\left( 1-
    \frac{1}{4n} \right)^n}{\mathcal{Z}_{s=1}}
\to \frac{e^{-1/4}}{e^{\sqrt{n/2}}}.
\end{equation}
at large $n$.

\section{Eigenvalue Gap Lower Bound for Example 1}
\label{gaproof}

Let $\gamma_1(s)$ be the eigenvalue gap of $H_1(s)$ as defined in
\eq{h1}. In this appendix we prove that $\min_{1/2 \leq s \leq 1}
\gamma_1(s)$ occurs at $s=\frac{1}{2}$. Note as $H_1(1/2) = H_0(1)$, the
eigenvalue gap of $H_1(1/2)$ can be obtained by substituting \eq{bval}
into \eq{rampgap} and expanding to lowest order in $1/n$, which yields
\begin{equation}
\label{gapval}
\gamma_1 (1/2) = \frac{1}{\sqrt{2n}} + O(n^{-3/2}).
\end{equation}

To prove that $\min_{1/2 \leq s \leq 1} \gamma_1(s)$ occurs at
$s=1/2$, we introduce the following lemma, which is physically
intuitive, and can be regarded as loosely analogous to Le Chatelier's
principle.

\begin{lemma} \label{lechatelier} Let $H(\alpha) = H_0 + \alpha V$ for
  any two Hermitian operators $H_0$ and $V$. Let
  $\ket{\psi_0(\alpha)}$ be the ground state of $H(\alpha)$, with
  energy $E_0(\alpha)$, which we assume to be nondegenerate. For any
  operator $M$ let 
  $\langle M \rangle_{\alpha} = \bra{\psi_0(\alpha)} M \ket{\psi_0(\alpha)}$.
  Then $\frac{d}{d \alpha} \langle V \rangle_{\alpha} \leq 0$ for all
  $\alpha$. Also, $\frac{d^2 E_0}{d \alpha^2} \leq 0$ for all
  $\alpha$.
\end{lemma}

\begin{proof}
By the variational principle,
\begin{equation}
E_0(\alpha) \leq \bra{\psi(\alpha_0)} H(\alpha) \ket{\psi(\alpha_0)}
\end{equation}
for any $\alpha_0$. Expanding this yields
\begin{equation}
\label{tangent}
E_0(\alpha) \leq \langle H_0 \rangle_{\alpha_0} + \alpha \langle V
\rangle_{\alpha_0}.
\end{equation}
By the Hellman-Feynman theorem
\begin{equation}
\label{FH}
\left. \frac{d E_0}{d\alpha} \right|_{\alpha_0} = \langle V \rangle_{\alpha_0}.
\end{equation}
Thus the righthand side of \eq{tangent} is identifiable as the tangent
line to $E_0(\alpha)$ at $\alpha_0$. The fact that $E_0(\alpha)$ lies
below its tangent line at every point implies
\begin{equation}
\label{concave}
\frac{d^2 E_0}{d \alpha^2} \leq 0.
\end{equation}
Taking a derivative of \eq{FH} yields
\begin{equation}
\label{penultimate}
\left. \frac{d^2 E_0}{d \alpha^2} \right|_{\alpha_0} = \left. 
\frac{d}{d \alpha} \langle V \rangle_\alpha \right|_{\alpha_0}.
\end{equation}
Together, \eq{penultimate} and \eq{concave} yield
\begin{equation}
\frac{d}{d\alpha} \langle V \rangle_\alpha \leq 0.
\end{equation}
\end{proof}

From \eq{pzero}, we find that the ground state of $H_0(s=1)$,
which is the ground state of $H_1(s=1/2)$ satisfies
\begin{equation}
| \braket{\psi_0(1/2)}{0\ldots0} |^2 > \frac{1}{2} \quad \forall n.
\end{equation} 
By Lemma \ref{lechatelier}, the amplitude in the all zeros state will
monotonically increase as $s$ is increased beyond $1/2$. Thus,
\begin{equation}
\label{bigzero}
| \braket{\psi_0(s)}{0\ldots0} |^2 > \frac{1}{2} \quad
\textrm{$\forall n$ and $\forall s \geq \frac{1}{2}$.}
\end{equation}
With \eq{bigzero} in hand, are now prepared to prove that the
eigenvalue gap $\gamma_1(s)$ of $H_1(s)$ monotonically increases for
$s \geq 1/2$.

Let $\ket{\psi_1}$ denote the first excited state of $H_1$. By the
Hellman-Feynman theorem
\begin{equation}
\frac{d \gamma_1}{ds} = \bra{\psi_1} \frac{dH}{ds} \ket{\psi_1} -
\bra{\psi_0} \frac{dH}{ds} \ket{\psi_0}.
\end{equation}
Substituting in the expression for $H_1(s)$ for $s \geq 1/2$, given by
\eq{h1}, \eq{schedule1}, and \eq{schedule2} yields
\begin{equation}
\label{minus}
\frac{d \gamma_1}{ds} = 2 a \left( |
  \braket{\psi_0}{0\ldots0}|^2 - | \braket{\psi_1}{0\ldots0}|^2
\right).
\end{equation}
By \eq{bigzero},  $|\braket{\psi_0}{0\ldots0}|^2 \geq 1/2$ for all
$s \geq 1/2$. By the orthogonality of $\bra{\psi_0}$ and
$\bra{\psi_1}$, this implies $|\braket{\psi_1}{0\ldots0}|^2
\leq 1/2$ for all $s \geq 1/2$. Consequently, \eq{minus} yields
\begin{equation}
\frac{d \gamma_1}{ds} \geq 0 \quad \forall s \geq 1/2,
\end{equation}
and therefore the minimum gap in this stage of the adiabatic process
occurs at $s=1/2$ and is as given in \eq{gapval}. (Note that gap at
$s=0$ is $2/n$, and this is the minimum gap for the whole adiabatic
process.)

\section{Permutation Symmetry and Efficient Spectrum Calculation}
\label{symmetry}

In this appendix, we prove that the eigenvalue gap of the $2^n \times
2^n$ matrix $H_1(s)$ is equal to the eigenvalue gap of the
$(n+1)\times(n+1)$ block of $H_1(s)$ that acts on the
permutation-symmetric subspace. Consequently, we can numerically
calculate both the eigenvalue gap of $H_1(s)$ and the ground state
eigenvector in $O(n^3)$ time. For example, the data for
Figure \ref{numerics}, which extends up to $n=2000$ was computed in
under an hour on a standard workstation. More generally, for
permutation-symmetric Hamiltonians on $n$ qubits, the degeneracies
ensure that there are only $\mathrm{poly}(n)$ distinct
eigenvalues. All of these can be computed in $\mathrm{poly}(n)$ time
using the methods outlined in the supplemental material of \cite{Boixo4}.

\begin{proposition}
Let $H$ be a Hamiltonian of the form
\begin{equation}
H_a = -c \sum_{j=1}^n X_j + b \sum_{j=1}^n Z_j - a
\ket{0\ldots0}\bra{0\ldots0}
\end{equation}
where $a \geq 0$. Then the eigenvalue gap of $H$ is the eigenvalue gap
of the $(n+1) \times (n+1)$ block of $H$ acting on the
permutation-symmetric subspace.
\end{proposition}

\begin{proof}
The ground state of $H_{a=0}$ is $\ket{\psi}^{\otimes n}$ where
$\ket{\psi}$ is the ground state of $-cX + bZ$. Similarly, the first
excited level of $H_{a=0}$ is $n$-fold degenerate and spanned by
\begin{equation}
\begin{array}{rcl}
\ket{\phi_1} & = & \ket{\bar{\psi}} \ket{\psi} \ket{\psi} \ldots \ket{\psi} \vspace{3pt} \\
\ket{\phi_2} & = & \ket{\psi} \ket{\bar{\psi}} \ket{\psi} \ldots \ket{\psi} \\
             & \vdots & \\
\ket{\phi_n} & = & \ket{\psi} \ket{\psi} \ket{\psi} \ldots \ket{\bar{\psi}}
\end{array}
\end{equation}
where $\ket{\bar{\psi}}$ is the excited state of $-cX + bZ$. Thus, one
sees that the ground state of $H_{a=0}$ is invariant under all
permutations of the qubits and the first excited eigenspace contains a
permutation-symmetric state, namely $\frac{1}{\sqrt{n}} \sum_{j=1}^n
\ket{\phi_n}$.

Now, recall some standard facts about the spin angular-momentum
operators.
\begin{equation}
\begin{array}{lclcl}
S_x = \frac{1}{2} \sum_{j=1}^n X_j & \quad &
S_y = \frac{1}{2} \sum_{j=1}^n Y_j & \quad &
S_z = \frac{1}{2} \sum_{j=1}^n Z_j \vspace{10pt}\\
S^2 = S_x^2 + S_y^2 + S_z^2 & & \end{array}
\end{equation}
The eigenvalues of $S^2$ are $j(j+1)$ for $j=0,1,2,\ldots,\frac{n}{2}$
if $n$ is even and $j = \frac{1}{2}, \frac{3}{2}, \frac{5}{2}, \ldots
\frac{n}{2}$ if $n$ is odd. $S_z$ commutes with $S^2$ and therefore
they can be simultaneously diagonalized. In the space with $S^2$
eigenvalue $j(j+1)$ (called the spin-$j$ space) the eigenvalues of
$S_z$ are $-j, -j+1,\ldots,j-1,j$. The spin-$\frac{n}{2}$ subspace is
precisely the permutation-symmetric subspace.

Examining $S_z$ one sees that this means the spin-$j$ space has
support only on bitstrings with Hamming weights $\frac{n}{2} - j,
\frac{n}{2} -j + 1, \ldots, \frac{n}{2} +j -1, \frac{n}{2} + j$. 
$S^2$ commutes with $H$. Thus, in an appropriate basis, $H$ is block
diagonal one block corresponding to each allowed value of $j$.
\begin{equation}
H_a = \begin{array}{c} \includegraphics[width=1.3in]{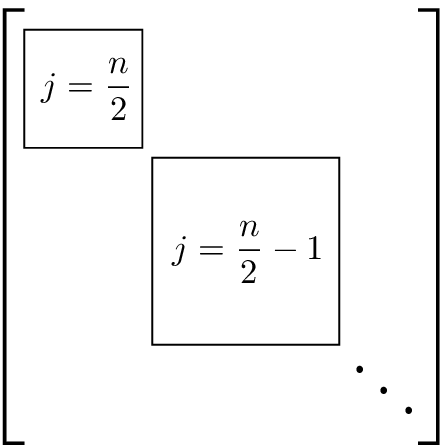} \end{array}
\end{equation}
Now imagine we start with $a=0$ and then increase $a$. The $j \neq
\frac{n}{2}$ blocks of $H$ will be completely unaffected since they
act on subspace that exclude Hamming weight zero. Only the
$j=\frac{n}{2}$ block can be affected. The eigenvalues of the operator
$-a \ket{0\ldots0}\bra{0\ldots0}$ are $-a$ and $0$. Thus, by Weyl's
inequality, adding this operator to the $j=\frac{n}{2}$ block will
lower each of its eigenvalues by some amounts between zero and $a$. The
ground energy of $H_{a=0}$ comes from the $j=\frac{n}{2}$ block, and
the first excited energy degenerately comes from the $j=\frac{n}{2}$
block. As we increase $a$ away from zero, the eigenvalues from this
block decrease (or remain constant), while the eigenvalues form the
other blocks are unchanged. Thus, the lowest two eigenvalues of $H$
will continue to be the lowest two eigenvalues of the $j=\frac{n}{2}$
block for all positive $a$.
\end{proof}

The fact that $\ket{\psi_s}$ is permutation-symmetric can be exploited
to numerically compute the ground state and ground energy up to large
numbers of qubits. Specifically, let $\ket{\psi_w}$ be the uniform
superposition over all length-$n$ bitstrings of Hamming weight $w$.
\begin{equation}
\ket{\psi_w} = \frac{1}{\sqrt{\binom{n}{w}}} \sum_{|x| = w} \ket{x}
\end{equation}
Then the ground state $\ket{\psi_s}$ can be expressed as
\begin{equation}
\ket{\psi_s} = \sum_{w=0}^n \alpha_w \ket{\psi_w}.
\end{equation}
The Hamiltonian $H_1(s)$ can be block-diagonalized with one block
corresponding to
$\mathrm{span}\{\ket{\psi_{w=0}}, \ldots, \ket{\psi_{w=n}} \}$. In
particular, a brief calculation yields 
\begin{equation}
- \sum_{j=1}^n X_j \ket{\psi_w} = - \sqrt{(w+1)(n-w)} \ket{\psi_{w+1}}
- \sqrt{w(n-w+1)} \ket{\psi_{w-1}}.
\end{equation}
Using this formula, one can easily write down an expression for the
Hamming-symmetric block of $H_1(s)$, which is a simple
$(n+1)\times(n+1)$ tri-diagonal matrix. Diagonalizing this matrix
yields the ground state eigenvalue and eigenvector, as well as the
rest of the permutation-symmetric part of the spectrum. From the
ground state eigenvector, we can calculate  $p_{s=1}^{(1)}$.

\bibliography{diffusion.bib}

\end{document}